\documentclass{article} 
\usepackage{iclr2026_conference,times}

\iclrfinalcopy

\usepackage{hyperref}
\usepackage{url}
\usepackage{amsmath,amssymb,amsthm}
\usepackage{mathtools}
\usepackage{algorithm}
\usepackage{algpseudocode}
\usepackage{booktabs}
\usepackage{subcaption}
\usepackage{graphicx}
\usepackage{xcolor}
\usepackage{multirow}

\newcommand{\R}{\mathbb{R}}
\newcommand{\E}{\mathbb{E}}
\newcommand{\Prob}{\mathbb{P}}
\newcommand{\Sspace}{\mathcal{S}}
\newcommand{\Aspace}{\mathcal{A}}
\newcommand{\Beps}{\mathcal{B}_{\varepsilon}}
\newcommand{\Otil}{\widetilde{\mathcal{O}}}

\newcommand{\SW}{\mathrm{SW}}
\newcommand{\NE}{\mathrm{NE}}
\newcommand{\PoA}{\mathrm{PoA}}
\newcommand{\PoP}{\mathrm{PoP}}
\newcommand{\CW}{\mathrm{CW}}
\newcommand{\pmax}{\pi_{\mathrm{mm}}}
\newcommand{\prat}{\pi_{\mathrm{rat}}}
\newcommand{\Wstar}{W^{*}}
\newcommand{\Wmm}{W_{\mathrm{mm}}}
\newcommand{\CVaR}{\text{CVaR}}
\newcommand{\VaR}{\text{VaR}}

\newcommand{\rc}{r_{\mathrm{c}}}
\newcommand{\rh}{r_{\mathrm{h}}}
\newcommand{\rs}{r_{\mathrm{s}}}
\newcommand{\Del}{\Delta}

\newcommand{\EVaR}{\mathrm{EVaR}}
\newcommand{\sigp}{\sigma^{2}_{\hat{p}}}
\newcommand{\trust}{\tau}
\newcommand{\phat}{\hat{p}}
\newcommand{\pstar}{p^{*}}
\newcommand{\bstar}{\beta^{*}}
\newcommand{\betat}{\beta(t)}

\DeclareMathOperator*{\argmax}{arg\,max}

\newcommand{\ind}[1]{\mathbf{1}\!\left[#1\right]}

\newcommand{\card}[1]{\left|#1\right|}

\newtheorem{assumption}{Assumption}
\newtheorem{theorem}{Theorem}
\newtheorem{proposition}[theorem]{Proposition}
\newtheorem{lemma}[theorem]{Lemma}
\newtheorem{corollary}[theorem]{Corollary}
\newtheorem{definition}{Definition}

\title{The Price of Paranoia: Robust Risk-Sensitive Cooperation in Non-Stationary Multi-Agent Reinforcement Learning}

\author{%
\textbf{Deep Kumar Ganguly}$^{\dagger, 1}$, \textbf{Chandradithya S Jonnalagadda}$^{\ddagger}$, \textbf{Pratham Chintamani}$^{\clubsuit}$, \textbf{Adithya Ananth}$^{\clubsuit}$ \\[0.5em]
$^{\dagger}$Technical University of Munich (TU Munich) \\
$^{\ddagger}$Brown University \\
$^{\clubsuit}$Indian Institute of Technology Tirupati (IIT Tirupati) \\
Correspondence to: \texttt{deep.ganguly@tum.de}$^{\dagger}$, \texttt{cjonnala@cs.brown.edu}$^{\ddagger}$
}

\begin{document}

\maketitle
\footnotetext[1]{Supported by the German Research Foundation (DFG) through Research Training Group GRK 2428 ConVeY.}
\lhead{Preprint}  

\begin{abstract}
Cooperative equilibria are fragile. When agents learn alongside each other rather than against a fixed environment, the very process of learning destabilizes the cooperation they are trying to sustain: every gradient step one agent takes shifts the distribution of actions its partner will play, turning a cooperative partner into a source of stochastic noise precisely where the cooperation decision is most sensitive. We study how this co-learning noise propagates through the structure of coordination games, and find that the cooperative equilibrium---even when strongly Pareto-dominant---is exponentially unstable under standard risk-neutral learning, collapsing irreversibly once partner noise crosses the game's critical cooperation threshold. The natural response, applying distributional robustness to hedge against partner uncertainty, makes things strictly worse: risk-averse return objectives penalize the high-variance cooperative action relative to the safe defection action, widening the instability region rather than shrinking it---a paradox that reveals a fundamental mismatch between the domain where robustness is applied and the domain where instability actually originates. We resolve this by showing that robustness should target the policy gradient update variance induced by partner uncertainty, not the return distribution itself. This distinction yields an algorithm whose gradient updates are continuously modulated by an online measure of partner unpredictability, provably expanding the cooperation basin in any symmetric coordination game. To unify the stability, sample complexity, and welfare consequences of this approach, we introduce the Price of Paranoia as the structural dual of the Price of Anarchy---a game-theoretic quantity that, together with a novel Cooperation Window, precisely characterizes how much welfare any learning algorithm can recover under partner noise, and pins down the optimal degree of robustness as a closed-form balance between equilibrium stability and sample efficiency.
\end{abstract}

\section{Introduction}
\label{sec:intro}

Large-scale cooperation among non-kin is a foundational pillar of human
societies---a primary catalyst for ecological and economic success
\citep{fehr2007human,handley2020human}.  Societies with robust cooperative
frameworks effectively manage shared resources, outcompeting less cohesive
groups and scaling in complexity and wealth \citep{boyd2009culture}.
Conversely, the erosion of cooperation is a precipitating factor in the
collapse of complex societies \citep{ferrarini2013economics}: free-riding
individuals trigger retaliatory cascades that unravel collective welfare
\citep{wardil2019positive,andreozzi2020stability,griessmair2022anger}.
Understanding not just how cooperation is established, but how it is \emph{sustained}, is therefore as consequential for artificial agents as it is for human societies.

Researchers model these dynamics via game theory and Multi-Agent Reinforcement Learning (MARL), which captures the temporal dynamics of co-learning and allows agents to adapt to the evolving policies of their peers. While modern MARL algorithms can converge to Pareto-optimal equilibria \citep{semsar2009multi,zhang2025optimism}, they routinely fail to \emph{sustain} cooperation over repeated interaction. This fragility stems from a source that is both obvious in hindsight and largely unaddressed in the literature: the agents themselves. Every gradient update one agent takes shifts the distribution of actions its partner will play, injecting stochastic noise into the cooperation signal at the exact moment the cooperation decision is most sensitive. Accidental defections during exploration are statistically indistinguishable from hostile policy shifts, provoking retaliation and triggering irreversible defection cascades \citep{10.1257/aer.101.1.411}.

\paragraph{The Optimist's Hangover.} In the canonical Stag Hunt \citep{skyrms2004stag}, partner exploration noise is statistically indistinguishable from strategic defection to a risk-neutral agent optimizing $\E[R]$. Even transient variance triggers a cascade of negative advantage signals, driving the cooperation probability below the critical threshold $\pstar$ and permanently locking the agent into the risk-dominant, suboptimal equilibrium. Optimism-based methods \citep{palmer2018lenient,zhang2025optimism} excel at discovering cooperation but remain hypersensitive to this variance. We term this the Optimist's Hangover: cooperation is learned optimistically, then lost paranoiacally.

\paragraph{The Adaptivity--Robustness Dilemma.} Traditional countermeasures fall short of resolving this sensitivity. Hysteresis \citep{matignon2007hysteretic,bowling2002multiagent} relies on hand-tuned dampening rather than calibrated uncertainty, while explicit opponent modelling \citep{albrecht2018autonomous,gmytrasiewicz2005framework,raileanu2018modeling} requires restrictive parameterized classes. Conversely, applying distributional robustness \citep{nilim2005robust,iyengar2005robust,zhang2020robustmarl,shi2023distributionally,kumar2024robust}---originally designed for adversarial settings---exposes a deeper tension: an agent that adapts too quickly treats partner exploration as defection, while one that is excessively robust suppresses genuine cooperation signals. This adaptivity--robustness dilemma remains fundamentally unresolved in non-stationary cooperative MARL.

\paragraph{The EVaR Paradox and Its Resolution.} Attempting to resolve this dilemma via standard distributional robustness actually exacerbates the Hangover. We prove that applying Entropic Value-at-Risk (EVaR) \citep{ahmadijevid2012entropic} directly to action-conditioned returns strictly increases the critical cooperation threshold for all $\beta > 0$ (where $\beta$ is the risk-sensitivity parameter, with $\beta>0$ corresponding to risk-aversion; see~\S\ref{sec:prelim}). This EVaR Paradox occurs because return-level risk-aversion disproportionately penalizes the high-variance cooperative action. The paradox resolves, however, when EVaR is applied instead to the policy gradient update variance induced by partner uncertainty. This formulation yields a closed-form trust factor $\trust(t) = (1 + \betat\,\sigp(t))^{-1}$, where $\sigp(t)$ is the Bernoulli variance of an online partner model. Crucially, this mechanism is doubly adaptive: the trust factor continuously modulates gradient updates, while the risk parameter $\betat$ adapts online by tracking our novel welfare diagnostic, the Price of Paranoia (PoP). Ultimately, agents do not need prosocial priors \citep{peysakhovich2017prosocial} or costly punishment mechanisms \citep{10.1257/aer.90.4.980,NBERw21457} to sustain cooperation; they simply require calibrated uncertainty---a state of principled paranoia.

\paragraph{The Price of Paranoia.}
We introduce the PoP as the structural dual of the Price of Anarchy
\citep{papadimitriou2005computing,roughgarden2015intrinsic}: where PoA quantifies the welfare cost of rational selfishness from below, PoP quantifies the welfare cost of maximin conservatism from above. Together they bracket the achievable welfare space via a \emph{Cooperation Window} $\CW(G,\varepsilon)$, which pins down the optimal risk parameter $\bstar$ as a closed-form balance between equilibrium stability and sample-complexity overhead. The dynamic variant $\PoP_{\mathrm{DYN}}(t)$ provides a normalized, game-comparable welfare diagnostic---analogous to adaptive regret \citep{blum2008regret}---measuring in real time how completely an adaptive agent recovers cooperative welfare after a perturbation.

\paragraph{Contributions.} We establish the theoretical foundations for risk-aware cooperative MARL and operationalize them through \textbf{Robust Adaptive Trust-Region Learning (RATTL)}, the first doubly-adaptive cooperative MARL algorithm. Our specific contributions are:
\begin{itemize}
\item \emph{The EVaR Paradox} (Proposition~\ref{prop:paradox}): we prove that applying EVaR directly to return distributions counterintuitively hinders cooperation by widening the basin of instability.
\item \emph{Adaptive trust factor} (Definition~\ref{def:trust}): a zero-communication, closed-form gradient-variance robustification $\trust(t)=(1+\betat\sigp(t))^{-1}$ that resolves the paradox.
\item \emph{Basin-expansion theorem} (Theorem~\ref{thm:basin}): a formal proof that RATTL lowers the critical cooperation threshold without altruistic priors, accompanied by a PAC sample-complexity bound (Theorem~\ref{thm:pac}) with polynomial robustness overhead $\mathcal{O}(\card{\Aspace_j}e^\beta)$.
\item \emph{Price of Paranoia framework}: the Cooperation Window $\CW(G,\varepsilon)$ and dimensionally corrected optimal risk formula $\bstar$ unifying stability, sample complexity, and welfare.
\item \emph{Online adaptation} (Algorithm~\ref{alg:adaptive_beta}): an adaptive $\betat$ rule tracking $\PoP_{\mathrm{DYN}}(t)$ in real time.
\item \emph{Empirical validation}: RATTL retains near-100\% cooperation under severe partner noise where risk-neutral and prosocial baselines collapse, with a transparent failure analysis under extreme non-stationarity that motivates the adaptive $\betat$ rule.
\end{itemize}

\section{Related Work}
\label{sec:related}
\paragraph{Cooperation mechanisms.}
The tension between individual incentives and collective welfare is traditionally addressed via external sanctioning \citep{fehr2000cooperation, 10.1257/aer.90.4.980, NBERw21457}, intrinsic motivation models like guilt \citep{griessmair2022anger}, or reward-shaping mechanisms such as inequity aversion \citep{hughes2018inequity} and prosocial reward mixing \citep{wang2021gifting, peysakhovich2017prosocial, mckee2020social}. These methods require auxiliary resources, complex credit assignment, or privileged access to partner utilities. RATTL modifies only the gradient update, preserving the original environment and communication protocols.

\paragraph{Opponent modelling and shaping.}
Explicitly modelling opponents \citep{albrecht2018autonomous, gmytrasiewicz2005framework, raileanu2018modeling} or differentiating through their learning steps, as in LOLA \citep{foerster2018learning}, directly addresses co-learning noise by accounting for how a partner's policy changes in response to one's own update. However, LOLA and RATTL operate on fundamentally different axes of the problem. LOLA requires access to the partner's policy parameters and differentiable learning rule to compute second-order gradient corrections; under high non-stationarity, the estimated partner response itself becomes noisy, potentially amplifying rather than dampening instability. RATTL instead targets the agent's own gradient variance induced by partner uncertainty, using only observed partner actions via a scalar EMA---no access to $\theta_j$, $R_j$, or the partner's learning algorithm is needed. The two approaches are complementary: LOLA shapes the partner toward cooperation (an active mechanism), while RATTL makes the agent's own updates robust to partner noise (a passive mechanism). RATTL's advantage is minimal information requirements and $\mathcal{O}(1)$ computational overhead per update; its limitation is that it cannot actively influence partner behavior.

\paragraph{Equilibrium selection.}
Equilibrium selection in MARL historically struggles with Pareto-optimal coordination in risk-dominant social dilemmas \citep{claus1998dynamics, skyrms2004stag, harsanyi1988general}. Hysteretic \citep{matignon2007hysteretic, palmer2018lenient, omidshafiei2017deep} and optimistic methods \citep{zhang2025optimism} mitigate early defection by filtering negative updates, excelling at cooperation discovery but collapsing under sustained partner variance. RATTL complements these approaches by targeting cooperation \emph{retention} rather than discovery.

\paragraph{Risk-sensitive and robust MARL.}
While CVaR and EVaR have been explored in single-agent RL \citep{tamar2015optimizing, chow2017risk, greenberg2022efficient}, and distributionally robust MARL tackles adversarial uncertainty \citep{nilim2005robust, iyengar2005robust, panaganti2022sample, zhang2020robustmarl, kumar2024robust, shi2023distributionally}, our work uniquely applies EVaR to gradient variance rather than the return distribution, and formally establishes that importing maximin adversarial methods into cooperative settings is counterproductive (Proposition~\ref{prop:paradox}).

\paragraph{Welfare theory.}
Our framework builds upon the Price of Anarchy \citep{papadimitriou2005computing, roughgarden2015intrinsic, awerbuch2005price} and adaptive regret dynamics \citep{blum2008regret}. By unifying these with online hyperparameter adaptation \citep{xu2018meta, zheng2018learning}, we introduce $\PoP_{\mathrm{DYN}}(t)$, a game-theoretically grounded diagnostic for non-stationary cooperative MARL, accompanied by a theoretically derived $\betat$ rule that balances equilibrium stability and sample complexity in real time.

\section{Preliminaries}
\label{sec:prelim}

\paragraph{Markov Games.}
A two-player Markov game
\[
  \mathcal{M} = (\Sspace,\;\Aspace_1,\;\Aspace_2,\;P,\;R_1,\;R_2,\;\gamma)
\]
has transition kernel~$P$, bounded rewards $R_i\in[R_{\min},R_{\max}]$, discount $\gamma\in[0,1)$, and parameterized policies $\pi_{\theta_i}\colon\Sspace\to\Delta(\Aspace_i)$. Each agent maximises $J_i(\theta_i;\pi_j) = \E[\sum_{t\geq 0}\gamma^t R_i(s_t,a_t^1,a_t^2)]$. We study stateless repeated games ($|\Sspace|=1$), so all quantities reduce to per-episode expectations over joint action distributions. Agent~$i$'s cooperation probability is $p = \pi_{\theta_i}(S)$, where $S$ denotes the cooperative action (Stag); agent~$j$'s is $q = \pi_{\theta_j}(S)$.

\begin{assumption}[Markov Game]\label{ass:mg}
$\mathcal{M} = (\Sspace, \Aspace_1, \Aspace_2, P, R_1, R_2, \gamma)$ with $P:\Sspace\times\Aspace_1\times\Aspace_2\to\Delta(\Sspace)$, $R_i:\Sspace\times\Aspace_1\times\Aspace_2\to[R_{\min}, R_{\max}]$, $\gamma\in[0,1)$. Agent $i$ has a parameterized policy $\pi_{\theta_i}:\Sspace\to\Delta(\Aspace_i)$.
\end{assumption}

\paragraph{Policy Gradient and Gradient Variance.}
Holding $\pi_j$ fixed, the policy gradient is: $\nabla_{\theta_i} J_i = \E[A_i \nabla_{\theta_i}\log\pi_{\theta_i}(a^i)]$, where $A_i = R_i - b$ is the advantage over baseline $b$. The REINFORCE update \citep{williams1992simple} implements stochastic gradient ascent; PPO \citep{schulman2017proximal} stabilizes it with a clipped surrogate, but the gradient structure is unchanged. For the Stag action $a^i = S$ with partner cooperation probability $q$, the partner-induced gradient variance is:
\begin{equation}\label{eq:grad_var}
  \Sigma(q) = q(1-q)\,\Del^2\,\bigl\|\nabla_{\theta_i}\log\pi_{\theta_i}(S)\bigr\|^2
\end{equation}
where $\Del = \rc - \rs$ is the payoff spread. This variance is non-negligible at the cooperation threshold $q = \pstar$ and is the direct target of RATTL's trust factor.

\paragraph{Coordination Games and the Stag Hunt.}
A symmetric game $G = (\Aspace, R)$ is a \emph{coordination game} if it has a payoff-dominant Nash equilibrium $\NE^c$ and a risk-dominant Nash equilibrium $\NE^r \neq \NE^c$ with $\SW(\NE^c) > \SW(\NE^r)$ \citep{harsanyi1988general}. The \emph{Stag Hunt} \citep{skyrms2004stag} is the canonical instance with $\Aspace = \{S,H\}$ and payoffs $\rc > \rh > \rs$ (Eq.~\eqref{eq:payoff}). The \emph{critical cooperation threshold} $\pstar = (\rh - \rs)/\Del \in (0,1)$ is the unique partner belief $q$ at which agent $i$ is indifferent between $S$ and $H$. The set $(\pstar,1]$ is the \emph{Basin of Trust}; $[0,\pstar)$ is the \emph{Basin of Fear}. At threshold, $A_S(\pstar) = 0$ while $\Sigma(\pstar) > 0$, so the signal-to-noise ratio of the gradient update vanishes exactly where the cooperation decision is hardest---the structural origin of the Optimist's Hangover.

\begin{definition}[Symmetric Coordination Game]\label{def:cg}
A symmetric game $G=(\Aspace,R)$ with $R(a_i,a_j)=R(a_j,a_i)$ is a \emph{coordination game} if it has a payoff-dominant NE and a risk-dominant NE with strictly higher social welfare at the former.
\end{definition}

\paragraph{Coherent Risk Measures and EVaR.}
Let $\Prob$ denote the reference probability measure over trajectories. A functional $\rho:\mathcal{X}\to\R$ is a \emph{coherent risk measure} \citep{artzner1999coherent} if it is monotone, sub-additive, positively homogeneous, and translation invariant. The \emph{Conditional Value-at-Risk} \citep{rockafellar2000optimization} $\CVaR_\alpha(X) = \E[-X \mid -X \geq \VaR_\alpha(X)]$ is coherent and widely used in risk-sensitive RL \citep{tamar2015optimizing,chow2017risk}, but its tail bound is loose. The \emph{Entropic Value-at-Risk} \citep{ahmadijevid2012entropic} is the tightest coherent risk measure under Cram\'{e}r's large-deviation bound:
\begin{equation}\label{eq:evar}
  \EVaR_\beta(X) \coloneqq \inf_{z>0} \Bigl\{\tfrac{1}{z}\log\tfrac{\E[e^{-zX}]}{1-\beta}\Bigr\}
  = \sup_{\substack{\Prob'\ll\Prob\\ D_{\mathrm{KL}}(\Prob'\|\Prob)\leq\log\frac{1}{1-\beta}}} \E_{\Prob'}[-X]
\end{equation}
EVaR dominates CVaR: $\E[-X] \leq \CVaR_\alpha \leq \EVaR_\beta \leq \mathrm{ess\,sup}(-X)$. As $\beta\to 0$, $\EVaR_\beta \to \E[-X]$ (risk-neutral); as $\beta\to 1$, $\EVaR_\beta \to \mathrm{ess\,sup}(-X)$ (maximin). Under the Gaussian approximation of the $K$-step return, $\EVaR_\beta(-X) \approx -\E[X] + \beta\,\mathrm{Var}(X)^{1/2}$, so the \emph{robust value} of a return $X$ (positive=good) is:
\begin{equation}\label{eq:evar_gauss}
  \mathrm{RV}_\beta(X) \;:=\; -\EVaR_\beta(-X) \;\approx\; \E[X] - \beta\,\mathrm{Var}(X)^{1/2}.
\end{equation}

\paragraph{Welfare and the Cooperation Window.}
The \emph{maximin strategy} $\pmax = \argmax_p \min_q u_i(p,q)$ satisfies $\pmax = 0$ for the Stag Hunt (Hare guarantees $\rh$ regardless of partner), giving paranoia floor $\Wmm = \SW(\pmax,\pmax) = 2\rh$. The \emph{social optimum} is $\Wstar = \SW(1,1) = 2\rc$. The \emph{cooperation efficiency} $\eta(\pi) = (\SW(\pi) - \Wmm)/(\Wstar - \Wmm) \in [0,1]$ measures what fraction of the welfare surplus above paranoid play an algorithm captures. The \emph{Price of Anarchy} $\PoA(G) = \Wstar/\SW(\NE^{\mathrm{worst}}) = \rc/\rh$ bounds achievable welfare from below at Nash. The \emph{Cooperation Window} $\CW(G,\varepsilon) = \Wstar - \Wmm(\varepsilon)$ is the welfare surplus any algorithm can capture above paranoid play under non-stationarity $\varepsilon$.

\section{Problem Formulation}
\label{sec:formulation}

We now consolidate the above into the precise mathematical problem that RATTL is designed to solve. Every object introduced in \S\ref{sec:prelim} appears below in its operational role.

\subsection{Game Environment and the Cooperation Retention Problem}
\label{sec:game-env-and-prob}
We study a two-agent Markov game $\mathcal{M}$ (Assumption~\ref{ass:mg}) instantiated as a repeated Stag Hunt \citep{skyrms2004stag}. We consider abstract payoffs $\rc > \rh > \rs$ with spread $\Del = \rc - \rs$:
\begin{equation}\label{eq:payoff}
  R(a_i,a_j) = \begin{cases}
    \rc & a_i=a_j=S \\ \rs & a_i=S,\,a_j=H \\ \rh & a_i=H
  \end{cases}
\end{equation}
Let $p=\Prob(a_i=S)$, $q=\Prob(a_j=S)$. Social welfare: $\SW(p,q) = 2pq\rc + (p+q-2pq)\rs + (2-p-q)\rh$.

Both agents update policies via gradient ascent on their individual objectives $J_i(\theta_i; \pi_j)$. Agent $i$'s \emph{observation} at each episode $t$ consists solely of the joint action pair $(a_i^{(t)}, a_j^{(t)}) \in \Aspace \times \Aspace$ and its own reward $R_i(a_i^{(t)}, a_j^{(t)})$. Agent $i$ has no access to: the partner's policy parameters $\theta_j$; the partner's reward $R_j$; the partner's learning rule; or the non-stationarity radius $\varepsilon$. Agent $j$ operates under non-stationarity (Assumption~\ref{ass:nonstat}) with radius $\varepsilon \geq 0$. Agent $i$ maintains an online estimate of the partner's cooperation probability:
\begin{equation}\label{eq:phat}
  \phat(t+1) = (1-\alpha)\phat(t) + \alpha\,\ind{a_j^{(t)} = S}
\end{equation}
with EMA coefficient $\alpha \in (0,1)$. The Bernoulli variance of this estimate is $\sigp(t) = \phat(t)(1-\phat(t))$, which serves as a scalar proxy for the partner's current unpredictability.

\begin{assumption}[$\varepsilon$-Non-Stationarity]\label{ass:nonstat}
At each episode $t$, partner $j$'s policy is drawn from a Wasserstein ball $\Beps(\pi_j^{\mathrm{nom}}) = \{\pi:W_1(\pi,\pi_j^{\mathrm{nom}})\leq\varepsilon\}$.\footnote{The order-1 Wasserstein distance between distributions $\mu,\nu$ over a metric space $(X,d)$ is $W_1(\mu,\nu)=\inf_{\gamma\in\Gamma(\mu,\nu)}\E_{(x,y)\sim\gamma}[d(x,y)]$. For policies over a finite action space, $W_1$ reduces to total variation up to a constant.} The agent observes partner \emph{actions} but not the drawn policy.
\end{assumption}

\begin{definition}[Cooperation Retention]
\label{def:retention}
Agent $i$ \emph{retains cooperation} over horizon $T$ if $p(t) \geq \pstar$ for all $t \in [t_0, t_0 + T]$. Retention \emph{fails} at episode $\tau > t_0$ if $p(\tau) < \pstar$; by the collapse dynamics of Proposition~\ref{prop:collapse}, this is irreversible under risk-neutral gradient learning.
\end{definition}

\medskip
\noindent\textbf{Problem (Robust Cooperation Retention).} Given a repeated Stag Hunt $G = (\{S,H\}, R)$ with payoffs $\rc > \rh > \rs$, partner non-stationarity radius $\varepsilon \geq 0$, and learning horizon $T$, find a gradient update rule for agent $i$ that: (1)~expands the cooperation basin (Theorem~\ref{thm:basin}); (2)~guarantees retention with high probability under $\varepsilon$-non-stationarity (Theorem~\ref{thm:welfare}); (3)~admits sample-complexity bounds polynomial in the effective action-space dimension (Theorem~\ref{thm:pac}); and (4)~operates without privileged information.

\subsection{Why Adaptive Agents Lose Cooperation}
\label{sec:adaptive-agents-lose-cooperation}
\begin{proposition}\label{prop:threshold}
A risk-neutral agent plays Stag iff partner belief $q \geq \pstar$:
\begin{equation}\label{eq:pstar}
  \pstar = (\rh - \rs)/(\rc - \rs) = (\rh - \rs)/\Del
\end{equation}
\end{proposition}

\begin{lemma}[Variance near Threshold]\label{lem:var}
Under Bernoulli partner cooperation $q$: $\mathrm{Var}(R_S \mid q) = q(1-q)\Del^2$. At $q=\pstar$, the signal-to-noise ratio $\mathrm{SNR} = (\pstar\Del - (\rh-\rs))^2 / (\pstar(1-\pstar)\Del^2)$ vanishes, making the agent maximally sensitive to stochastic fluctuations exactly where the cooperation decision is made.
\end{lemma}

\begin{proposition}[Exponential Cooperation Collapse]\label{prop:collapse}
A risk-neutral agent at $q=\pstar+\epsilon$, $\epsilon>0$ small, under partner exploration rate $\delta\in(0,1)$ and policy-gradient learning rate $\eta_{\mathrm{lr}}>0$, satisfies:
\begin{equation}\label{eq:collapse}
  p(t) \approx \pstar + \epsilon\,e^{-\lambda\delta t}, \quad \lambda = \eta_{\mathrm{lr}}(\rh-\rs)\,/\,\Del
\end{equation}
Once $p(t) < \pstar$, the agent irreversibly collapses to Hare.
\end{proposition}

The collapse is \emph{irreversible} because below $\pstar$, the gradient of the cooperation action is always negative under risk-neutral learning. \emph{Adaptation fails precisely because the agent adapts too quickly}: it treats partner exploration as a signal about long-run partner intent rather than transient variance.

\subsection{The EVaR Paradox: Why Standard Robustness Fails}
\label{sec:paradox_formal}

A natural approach to retention is to replace the expected-return objective with an EVaR objective $\max_{\theta_i} -\EVaR_\beta(-\sum_{t \geq 0}\gamma^t R_i(a_i^{(t)}, a_j^{(t)}))$. Under the Gaussian approximation Eq.~\eqref{eq:evar_gauss}, this penalizes the Stag return variance and shifts the risk-adjusted threshold to the following.

\begin{proposition}[EVaR Paradox]\label{prop:paradox}
Under the Gaussian approximation~\eqref{eq:evar_gauss}, replacing the expected Stag return with its robust value $\mathrm{RV}_\beta(R_S\mid q) \approx Q_S(q) - \beta\sqrt{q(1-q)}\,\Del$ increases the critical cooperation threshold for all $\beta > 0$:
\begin{equation}\label{eq:paradox}
  {p^*_\beta}^{(\mathrm{naive})} = \pstar + \beta\sqrt{\pstar(1-\pstar)}/\Del > \pstar.
\end{equation}
\end{proposition}

\begin{proof}
Substitute $\mathrm{RV}_\beta(R_S\mid q) = Q_S(q) - \beta\sqrt{q(1-q)}\,\Del$ (from Eq.~\eqref{eq:evar_gauss} with $\mathrm{Var}(R_S\mid q) = q(1-q)\Del^2$ by Lemma~\ref{lem:var}) into the indifference condition $\mathrm{RV}_\beta(R_S\mid q^*) = \rh$ and linearize around $\pstar$.
\end{proof}

The cooperative basin $(\pstar_{\beta^{(\mathrm{naive})}}, 1]$ is \emph{strictly smaller} than the risk-neutral basin $(\pstar,1]$. Distributional robustness applied to rewards narrows the cooperation basin, accelerates Hangover collapse, and makes retention strictly harder for all $\beta > 0$. The paradox reveals that the domain of EVaR application determines the sign of its effect on cooperation. The correct domain is not the return distribution but the policy gradient update.

\section{RATTL: Robust Adaptive Trust-Region Learning}
\label{sec:method}

We propose RATTL, an algorithm that satisfies all four conditions of the Robust Cooperation Retention problem simultaneously via its trust factor. RATTL's key insight is to apply EVaR to the partner-induced gradient variance $\Sigma(q)$ from Eq.~\eqref{eq:grad_var}, not to the return distribution. Formally, the EVaR-regularised gradient update replaces the raw advantage $A_i$ with a trust-dampened advantage:
\begin{equation}\label{eq:damped_adv}
  \widetilde{A}_i(a^i;\, \phat, \beta) = \begin{cases}
      \trust(\sigp,\beta) \cdot A_i(a^i, a^j) & a^i = S \\
      A_i(a^i, a^j) & a^i = H
    \end{cases}
\end{equation}
where the trust factor $\trust : [0,\tfrac{1}{4}]\times\R \to \R_{>0}$ is:
\begin{equation}\label{eq:trust}
  \trust(\sigp, \beta) = (1 + \beta\,\sigp)^{-1}
\end{equation}

\paragraph{Sign convention.} Positive $\beta>0$ \emph{dampens} the Stag gradient ($\trust<1$), filtering partner noise---this is the risk-sensitive regime targeted by our basin-expansion guarantees. Negative $\beta<0$ \emph{amplifies} the Stag gradient ($\trust>1$), recovering risk-seeking optimism (cf.~\citealp{zhang2025optimism}); $\beta=0$ recovers vanilla policy gradient. The constraint $1+\beta\sigp>0$ (equivalently $\beta>-4$ since $\sigp\le\tfrac14$) ensures $\trust$ remains positive. Dampening applies only to the Stag gradient---the uncertain, noisy direction---leaving the Hare gradient unmodified.

\begin{definition}\label{def:trust}
Given online estimate $\phat(t)$ of partner cooperation probability with Bernoulli variance $\sigp(t) = \phat(t)(1-\phat(t))$, the \emph{adaptive trust factor} at episode $t$ is:
\begin{equation}\label{eq:trust_method}
  \trust(t) = \left(1 + \betat\,\sigp(t)\right)^{-1} \in (0, 1]
\end{equation}
\end{definition}

\paragraph{Why Only Stag?} The asymmetry in Eq.~\eqref{eq:damped_adv} is principled, not heuristic. The Hare return $\rh$ is independent of partner action, so $\mathrm{Var}(R_H \mid a^i=H) = 0$ and EVaR$_\beta(R_H) = \E[R_H] = \rh$ for all $\beta$. The partner-induced gradient variance is zero for Hare actions; dampening it would introduce bias without variance-reduction justification.

\paragraph{Risk-robustness trade-off.} Higher $\beta$ expands the cooperation basin (Theorem~\ref{thm:basin}) but increases sample complexity. We formalize this as a two-term welfare loss:
\begin{equation}\label{eq:tradeoff}
  \E[\mathrm{Loss}] \leq \underbrace{\CW(G,\varepsilon)\bigl(1 - e^{-\beta\varepsilon}\bigr)}_{\text{(i) equilibrium cost}} + \underbrace{c\,|\Aspace_j|\,e^{\beta}/T}_{\text{(ii) sample cost}}
\end{equation}
The optimal risk parameter $\bstar$ minimizes Eq.~\eqref{eq:tradeoff} (closed form in Corollary~\ref{cor:beta_star}).

\subsection{Adaptive $\betat$: Online Risk Calibration}

A fixed $\beta$ is suboptimal: during early training when the partner is erratic, higher $\beta$ is appropriate; once cooperation stabilizes, lower $\beta$ accelerates learning. We derive an online update rule from the $\PoP_{\mathrm{DYN}}$ diagnostic.

\begin{definition}[$\PoP_{\mathrm{DYN}}$ as Adaptivity Diagnostic]\label{def:popdyn_adapt}
The dynamic Price of Paranoia at episode $t$ is:
\begin{equation}\label{eq:popdyn}
  \PoP_{\mathrm{DYN}}(t) = \frac{\SW(\pi_t, \bar{q})}{\SW(\pmax, \bar{q})} \in [1,\, \PoP_{\mathrm{GT}}]
\end{equation}
$\PoP_{\mathrm{DYN}}(t) = 1$: agent is fully paranoid (cooperation lost). $\PoP_{\mathrm{DYN}}(t) = \PoP_{\mathrm{GT}}$: agent is fully cooperative.
\end{definition}

\begin{algorithm}[t]
\caption{RATTL with Adaptive $\betat$}
\label{alg:adaptive_beta}
\begin{algorithmic}[1]
\small
\Require Action $a$, reward $r$, partner action $a_j$
\Require Params: $\alpha$, $\eta$, $\eta_\beta$, $\bstar$ (estimated), $\underline\beta,\bar\beta$ with $\underline\beta\,\tfrac{1}{4}>-1$ and $\bar\beta < \varepsilon^{-1} - 1$
\State $\phat \gets (1{-}\alpha)\phat + \alpha\ind{a_j = S}$
\State $\sigp \gets \phat(1{-}\phat)$
\State $\widehat{\SW} \gets 2\,r_i$ \Comment{Welfare proxy (symmetric game)}
\State $\Delta\PoP \gets \widehat{\SW} - \widehat{\SW}_{\mathrm{prev}}$
\If{$\Delta\PoP < 0$} \Comment{Cooperation deteriorating}
  \State $\beta \gets \min(\beta + \eta_\beta\,\card{\Delta\PoP},\; \bar\beta)$
\Else \Comment{Recovering toward $\bstar$}
  \State $\beta \gets \beta - \eta_\beta(\beta - \bstar)_+$
\EndIf
\State $\trust \gets (1 + \beta\,\sigp)^{-1}$
\State $b \gets \mathrm{mean}(\mathcal{H})$;\quad $A \gets (r-b)\cdot\trust^{\ind{a=S}}$
\State $\theta \gets \theta + \eta\,A\,\nabla_\theta\log\pi_\theta(a)$
\State $\widehat{\SW}_{\mathrm{prev}} \gets \widehat{\SW}$
\State \Return $\theta,\,\beta$
\end{algorithmic}
\end{algorithm}

\subsection{The Unified Price of Paranoia Framework}

\begin{definition}\label{def:popgt}
\textbf{$\PoP_{\mathrm{GT}}$ (Game-Theoretic):} Structural welfare ceiling: $\PoP_{\mathrm{GT}}(G,\varepsilon) = \SW(\prat(\varepsilon),\bar{q}_\varepsilon)\,/\,\SW(\pmax,\bar{q}_\varepsilon)$. \\
\textbf{$\PoP_{\mathrm{ALG}}$ (Algorithmic):} Sample-complexity overhead: $\PoP_{\mathrm{ALG}}(\mathcal{A},\beta) = N_{\mathrm{robust}}(\beta)\,/\,N_{\mathrm{standard}} \leq \mathcal{O}(\card{\Aspace_j}e^\beta)$. \\
\textbf{$\PoP_{\mathrm{DYN}}(t)$ (Dynamic):} Adaptive diagnostic (Def.~\ref{def:popdyn_adapt}). \\
\textbf{Separation:} $\PoP_{\mathrm{GT}}$ is a property of the \emph{game}, $\PoP_{\mathrm{ALG}}$ of the \emph{algorithm}, and $\PoP_{\mathrm{DYN}}$ of the \emph{learning trajectory}.
\end{definition}

\begin{theorem}\label{thm:cw}
Define the \emph{Cooperation Window}: $\CW(G,\varepsilon) = \Wstar - \Wmm(\varepsilon) = \Wstar(1 - 1/\PoP_{\mathrm{GT}}(G,\varepsilon))$. Any algorithm achieves welfare $W(\pi) \in [\Wstar/\PoA,\, \Wstar]$. The dynamic diagnostic $\PoP_{\mathrm{DYN}}(t)$ is monotone in $\eta(\pi_t)$, with $\PoP_{\mathrm{DYN}}(t)=1$ iff $\eta(\pi_t)=0$ (full collapse to maximin) and $\PoP_{\mathrm{DYN}}(t)=\PoP_{\mathrm{GT}}$ iff $\eta(\pi_t)=1$ (full cooperative welfare).
\end{theorem}

\section{Theoretical Guarantees}
\label{sec:theory}

\subsection{Adaptive Basin Expansion}

\begin{theorem}\label{thm:basin}
For RATTL with adaptive $\betat$, payoff spread $\Del=\rc-\rs$, and risk-neutral threshold $\pstar=(\rh-\rs)/\Del$, the effective cooperation threshold at episode $t$ satisfies:
\begin{equation}\label{eq:basin}
  \pstar_{\betat} = \pstar - \frac{\betat\,\pstar(1-\pstar)}{\Del} + \mathcal{O}((\betat)^2(\pstar)^2/\Del^2)
\end{equation}
The expansion $\Delta\pstar = \betat\pstar(1-\pstar)/\Del$ is maximised at $\pstar = \tfrac{1}{2}$.
\end{theorem}

\begin{proof}[Proof Sketch]
The trust-dampened cooperation condition is $b + \trust(q)(Q_S(q)-b) \geq \rh$. Taylor-expanding at $q=\pstar$ with $Q_S(\pstar)=\rh$, using $Q_S'=\Del$ and $\trust'(q)|_{\pstar} = -\betat(1-2\pstar)(1+\betat\sigp)^{-2}$, the first-order perturbation gives $\delta q = -\betat\pstar(1-\pstar)/\Del$.
\end{proof}

\begin{corollary}\label{cor:selfcorrect}
When the adaptive update rule (Algorithm~\ref{alg:adaptive_beta}) detects $\Delta\PoP < 0$ and increases $\betat$ by $\eta_\beta\card{\Delta\PoP}$, the cooperation basin expands by $\delta(\Delta\pstar) = \eta_\beta\card{\Delta\PoP}\pstar(1-\pstar)/\Del > 0$.
\end{corollary}

\subsection{Robustness Guarantees}

\begin{theorem}[PAC Sample Complexity]\label{thm:pac}
Under Assumptions~\ref{ass:mg}--\ref{ass:nonstat}, for any accuracy $\epsilon_{\mathrm{PAC}}>0$ and confidence $\delta\in(0,1)$, RATTL returns an $\epsilon_{\mathrm{PAC}}$-optimal robust policy with probability at least $1-\delta$ using at most
\begin{equation}\label{eq:pac}
  N_{\mathrm{RATTL}} = \Otil\!\!\left(\frac{\card{\Sspace}^2\card{\Aspace}H^4}{\epsilon_{\mathrm{PAC}}^2} \cdot \card{\Aspace_j}e^\beta \cdot \log\tfrac{1}{\delta}\right)
\end{equation}
episodes (the accuracy parameter $\epsilon_{\mathrm{PAC}}$ is distinct from the non-stationarity radius~$\varepsilon$ of Assumption~\ref{ass:nonstat}). The robustness overhead is $\PoP_{\mathrm{ALG}} \leq \mathcal{O}(\card{\Aspace_j}e^\beta)$---polynomial in $\card{\Aspace_j}$, not exponential in the full state-action space.
\end{theorem}

\begin{theorem}[Welfare Loss Decomposition]\label{thm:welfare}
Expected per-episode welfare loss decomposes into two adaptive-robustness trade-off terms:
\begin{equation}\label{eq:wl}
  \E[\mathrm{Loss}] \leq \underbrace{\CW(G,\varepsilon)(1-e^{-\beta\varepsilon})}_{\text{(i) adaptation cost}} + \underbrace{c\,\card{\Aspace_j}e^\beta/T}_{\text{(ii) robustness overhead}}
\end{equation}
\end{theorem}

\begin{corollary}\label{cor:beta_star}
The optimal fixed risk parameter is:
\begin{equation}\label{eq:beta_star}
  \bstar = \frac{1}{1+\varepsilon} \log\!\left(\frac{\CW(G,\varepsilon)\cdot\varepsilon\cdot T}{c_{\mathrm{eff}}}\right)
\end{equation}
where $c_{\mathrm{eff}} = c\card{\Aspace_j}$ calibrated empirically.
\end{corollary}

\section{Experiments}
\label{sec:experiments}

We first demonstrate that our theoretical framework generalizes across coordination games via predicted threshold analysis (\S\ref{sec:game_generalization}), and then validate RATTL empirically on the Iterated Stag Hunt (\S\ref{sec:ish_exp}).

\subsection{Theoretical Generalisation Across Games}
\label{sec:game_generalization}

While our empirical evaluation focuses on the Stag Hunt, Theorem~\ref{thm:basin} applies to any symmetric coordination game satisfying Definition~\ref{def:cg}. Table~\ref{tab:game_gen} instantiates the basin-expansion prediction $\pstar_\beta = \pstar - \beta\pstar(1{-}\pstar)/\Del$ for three canonical games at $\beta=1.0$.

\begin{table}[t]
\centering
\caption{Predicted basin expansion across coordination games. $\pstar$: risk-neutral threshold; $\pstar_{\beta=1}$: RATTL threshold (Theorem~\ref{thm:basin}). Expansion denotes the percentage increase in the cooperation basin $[p^*,1]$.}
\label{tab:game_gen}
\begin{tabular}{@{}lcccc@{}}
\toprule
Game & $(\rc,\rh,\rs)$ & $\pstar$ & $\pstar_{\beta=1}$ & Expansion \\
\midrule
Stag Hunt & $(5,2,{-}5)$ & 0.70 & 0.58 & +17\% \\
Chicken   & $(4,2,{-}1)$ & 0.60 & 0.45 & +25\% \\
Pure Coord. & $(3,1,0)$    & 0.33 & 0.26 & +21\% \\
\bottomrule
\end{tabular}
\end{table}

The Chicken game has a smaller payoff spread $\Del=5$ vs.\ $\Del=10$ for the Stag Hunt, making the basin more sensitive to trust-factor dampening. The pure coordination game ($\rs=0$) has the lowest threshold and hence the widest initial basin; RATTL still provides a meaningful expansion.

\paragraph{Baseline positioning.}
Our empirical evaluation uses vanilla PPO as the risk-neutral baseline. Hysteretic Q-learning \citep{matignon2007hysteretic} and lenient learning \citep{palmer2018lenient} filter negative updates to aid cooperation \emph{discovery}, but they offer no principled mechanism for cooperation \emph{retention}---the regime RATTL targets. LOLA \citep{foerster2018learning} addresses co-learning noise via second-order gradient corrections but requires differentiable access to the partner's learning rule. RATTL's $\mathcal{O}(1)$ scalar trust factor achieves retention guarantees (Theorem~\ref{thm:basin}) under strictly weaker information assumptions.

\subsection{Iterated Stag Hunt}
\label{sec:ish_exp}

We train RATTL-PPO in the Iterated Stag Hunt against a stochastic opponent whose mixed strategy is sampled from a standard normal distribution and projected onto the simplex at each timestep. Training runs for 3000 episodes under two configurations: $\beta = -1.0$ (risk-averse, amplifies gradient signals) and $\beta = 1.0$ (risk-seeking, dampens noisy gradients via the trust factor). We compare against vanilla PPO \citep{schulman2017proximal}.

\begin{figure}[t]
     \centering
     \begin{subfigure}[b]{0.48\textwidth}
         \centering
         \includegraphics[width=\textwidth]{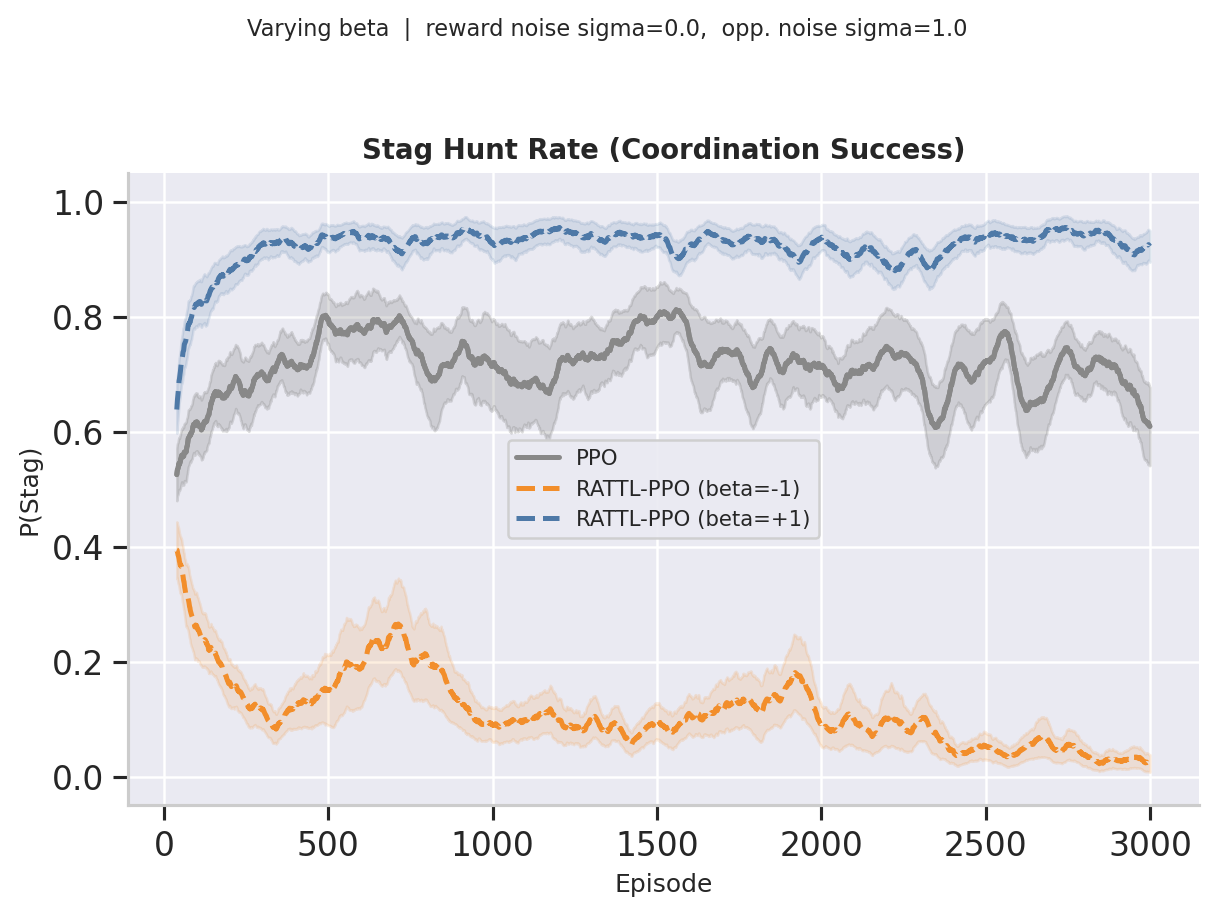}
         \caption{Cooperation rate (no reward noise)}
         \label{fig:vary_beta_rstd0.0_nstd1.0_stag}
     \end{subfigure}
     \hfill
     \begin{subfigure}[b]{0.48\textwidth}
         \centering
         \includegraphics[width=\textwidth]{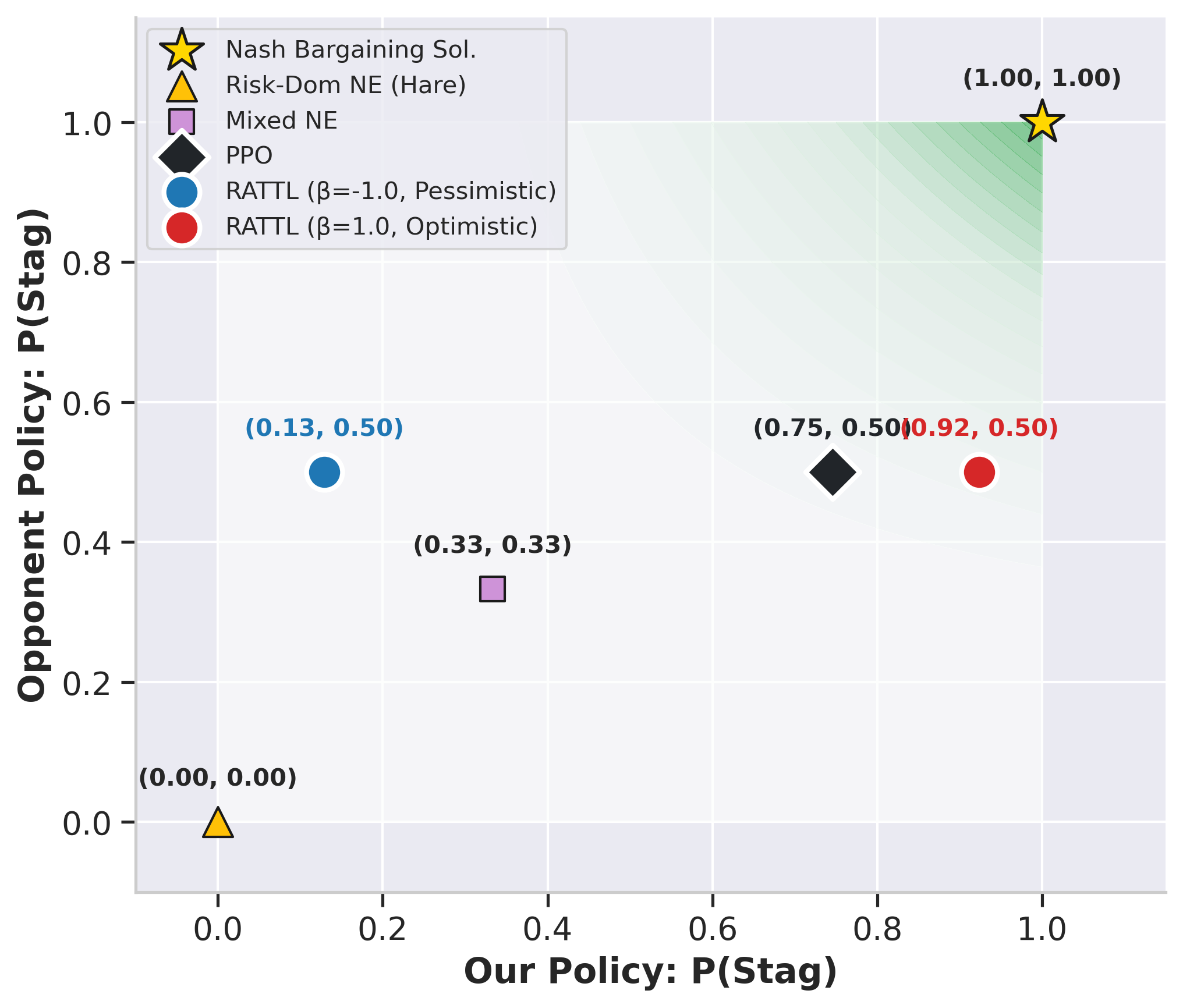}
        \caption{Outcome space (no reward noise)}
        \label{fig:vary_beta_outcome_space_rstd_0.0_oppstd_1.0}
     \end{subfigure}

     \vspace{0.3em}
     \begin{subfigure}[b]{0.48\textwidth}
         \centering
         \includegraphics[width=\textwidth]{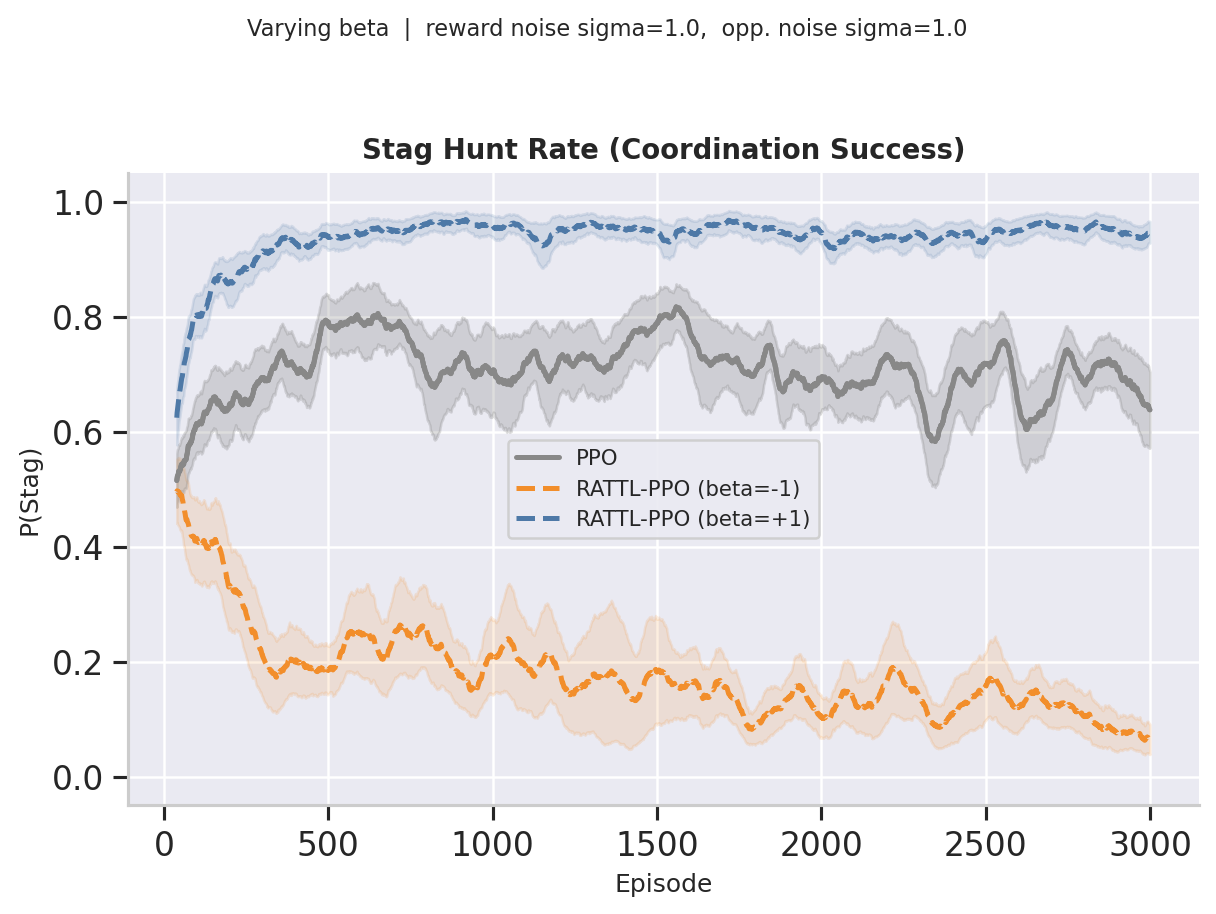}
         \caption{Cooperation rate (with reward noise)}
         \label{fig:vary_beta_rstd1.0_nstd1.0_stag}
     \end{subfigure}
     \hfill
     \begin{subfigure}[b]{0.48\textwidth}
         \centering
         \includegraphics[width=\textwidth]{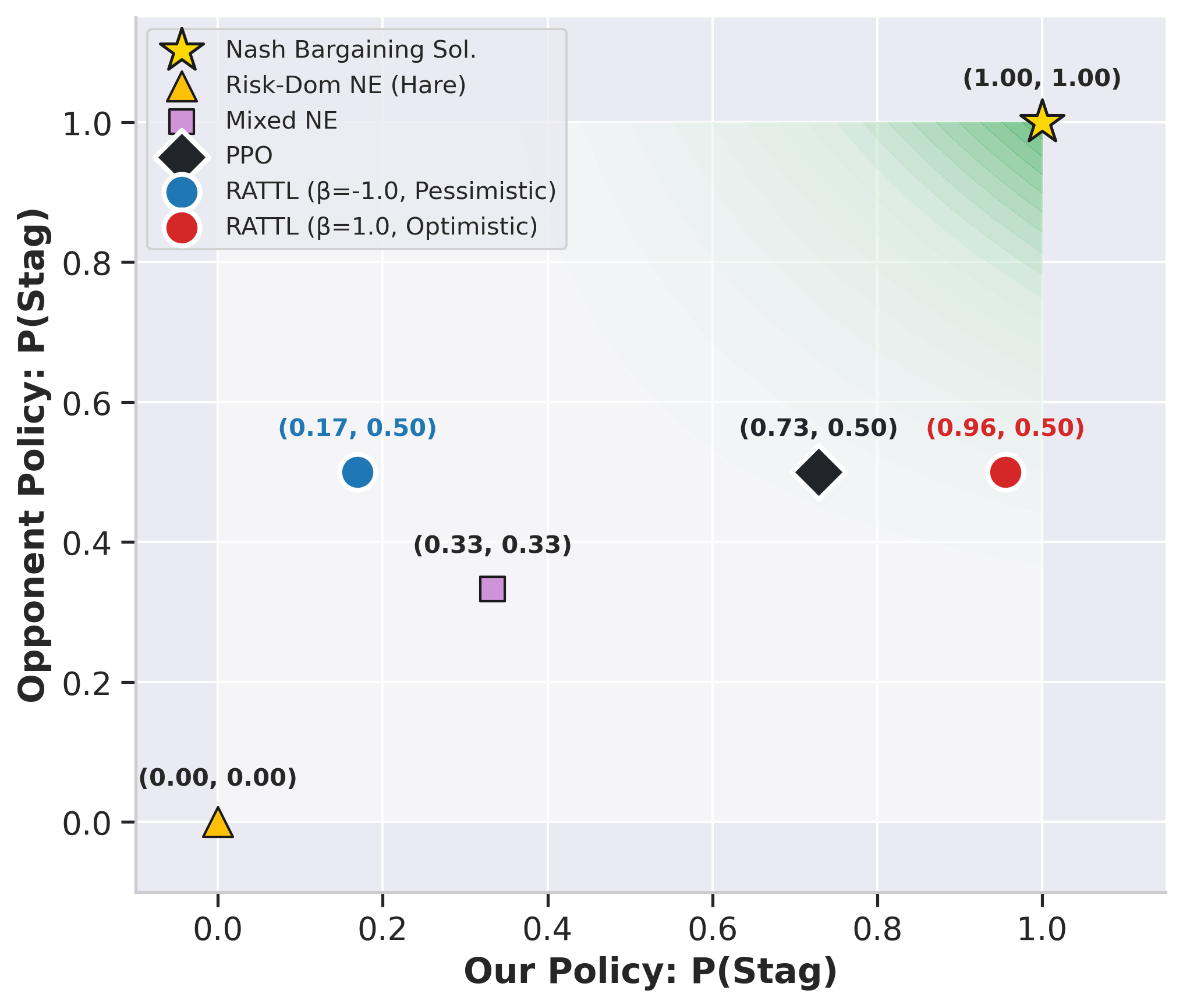}
        \caption{Outcome space (with reward noise)}
        \label{fig:vary_beta_outcome_space_rstd_1.0_oppstd_1.0}
     \end{subfigure}
     \caption{RATTL-PPO in the Iterated Stag Hunt. \textbf{Top row:} no reward noise. \textbf{Bottom row:} standard-normal reward perturbations. Risk-seeking RATTL ($\beta=1.0$) converges to stable cooperation near the NBS in both conditions; risk-averse ($\beta=-1.0$) defaults to Hare; vanilla PPO oscillates. Policies are projected against the Nash Bargaining Solution (NBS), risk-dominant NE, and mixed-strategy NE.}
     \label{fig:vary_beta_rstd0.0_oppstd_1.0}
\end{figure}

As shown in Figure~\ref{fig:vary_beta_rstd0.0_nstd1.0_stag}, risk-seeking RATTL-PPO ($\beta = 1.0$) establishes and maintains stable cooperation, while the risk-averse variant ($\beta = -1.0$) converges to Hare, and vanilla PPO settles on a volatile mixed strategy (${\approx}63\%$ Stag). Figure~\ref{fig:vary_beta_outcome_space_rstd_0.0_oppstd_1.0} projects policies into the outcome space against the Nash Bargaining Solution (NBS)---the unique Pareto-optimal point maximizing surplus gains over the disagreement point \citep{harsanyi1988general}. Risk-seeking RATTL reliably isolates the cooperative equilibrium near the NBS. The bottom row of Figure~\ref{fig:vary_beta_rstd0.0_oppstd_1.0} confirms that all configurations are robust to standard normal reward perturbations. Table~\ref{tab:ish_metrics} reports PoP and PoA under both conditions: risk-seeking RATTL ($\beta = 1.0$) dominates both metrics regardless of reward noise. Figure~\ref{fig:exp_criteria} (Appendix~\ref{app:extended}) further sweeps five risk profiles across stationary, mild ($20\%$), and high ($40\%$) partner noise: risk-aversion ($\beta\!\in\!\{1,2\}$) collapses sharply under high non-stationarity (final cooperation drops to ${\sim}0.5$), while risk-seeking and risk-neutral retain $>0.9$. The Pareto frontier (cooperation vs.\ stability) confirms that risk-seeking dominates under heavy noise---directly evidencing the EVaR Paradox (Proposition~\ref{prop:paradox}).

\begin{table}[t]
    \centering
    \caption{PoP and PoA in the Iterated Stag Hunt, with and without reward noise. Risk-seeking RATTL ($\beta=1.0$) dominates in both conditions.}
    \label{tab:ish_metrics}
    \begin{tabular}{@{}lcccc@{}}
        \toprule
        & \multicolumn{2}{c}{No noise} & \multicolumn{2}{c}{Reward noise} \\
        \cmidrule(lr){2-3}\cmidrule(lr){4-5}
        Algorithm & PoP$\uparrow$ & PoA$\downarrow$ & PoP$\uparrow$ & PoA$\downarrow$ \\
        \midrule
        RATTL ($\beta{=}{-}1$) & 1.27 & 4.23 & 1.34 & 3.96 \\
        RATTL ($\beta{=}1$) & \textbf{2.88} & \textbf{1.85} & \textbf{2.93} & \textbf{1.82} \\
        PPO & 2.51 & 2.13 & 2.49 & 2.14 \\
        \bottomrule
    \end{tabular}
\end{table}

\section{Conclusion}

We address the collapse of cooperation in non-stationary MARL driven by co-learning partner stochasticity. The EVaR Paradox (Proposition~\ref{prop:paradox}) shows that standard distributional robustness applied to returns widens the basin of instability; resolving this requires targeting gradient variance instead. This yields RATTL, whose adaptive trust factor $\trust(t)$ provably expands the cooperation basin and is unified under the Price of Paranoia framework. Empirically, risk-seeking RATTL achieves near-100\% cooperation retention where standard baselines collapse. Our framework demonstrates that stable cooperation requires neither opponent modelling nor prosocial priors---only calibrated uncertainty. Future directions include adaptive $\betat$ meta-learning, scalable deep RL implementations, and human-AI cooperation studies.

\subsubsection*{Acknowledgments}
This work has been supported by the German Research Foundation (DFG) through Research Training Group GRK 2428 ConVeY.

\bibliography{rattl_references}
\bibliographystyle{iclr2026_conference}

\appendix

\section{Extended Experimental Results}
\label{app:extended}

\begin{figure}[h]
    \centering
    \includegraphics[width=\linewidth]{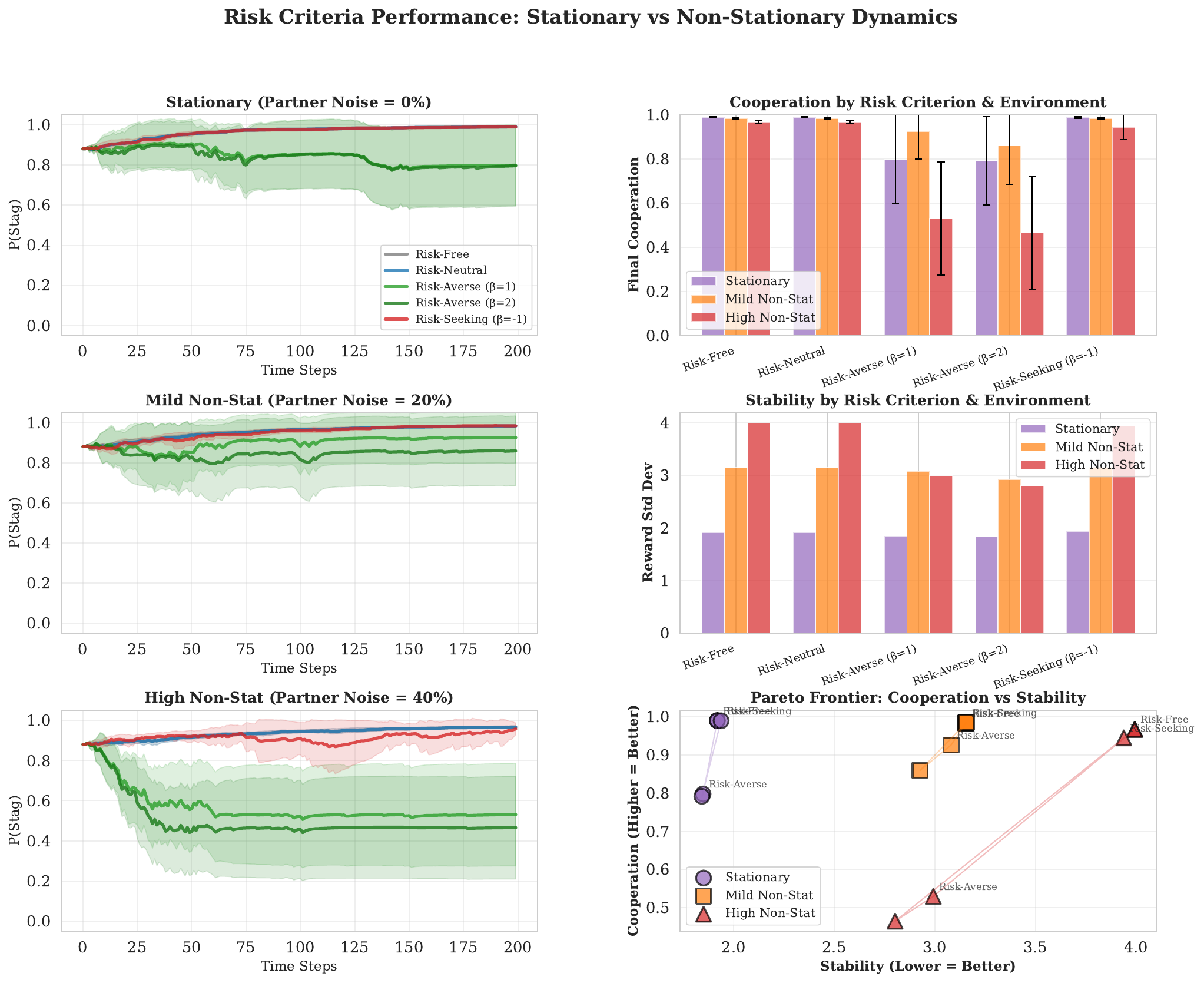}
    \caption{Risk-criteria performance across stationary and non-stationary partner dynamics. \textbf{Left column:} cooperation probability $P(\mathrm{Stag})$ over 200 episodes for five risk profiles (Risk-Free, Risk-Neutral, Risk-Averse $\beta\!\in\!\{1,2\}$, Risk-Seeking $\beta\!=\!-1$) under partner noise $\in\{0\%,20\%,40\%\}$. \textbf{Top-right:} final cooperation by criterion and environment; risk-aversion collapses sharply under high non-stationarity, while risk-seeking and risk-neutral retain $>0.9$. \textbf{Middle-right:} reward standard deviation. \textbf{Bottom-right:} Pareto frontier of cooperation vs.\ stability; risk-seeking dominates the upper-left under heavy noise. Empirically validates the EVaR Paradox (Proposition~\ref{prop:paradox}).}
    \label{fig:exp_criteria}
\end{figure}

We next inject standard normal noise ($\sigma \in \{0.5, 1.0\}$) directly into the opponent's mixed strategy, testing resilience to behavioral stochasticity beyond reward perturbations. All configurations show resilience to partner noise, though the nature of that resilience differs sharply across objectives.

\paragraph{Risk-seeking RATTL.} As shown in Figure~\ref{fig:vary_noise_beta1.0_rstd0.0_oppstd}, RATTL-PPO ($\beta = 1.0$) maintains cooperation (Stag rate $>92\%$) despite heavy partner noise (Table~\ref{tab:ish_metrics_opp_noise}: PoP peaks at 2.88, PoA drops to 1.85 at $\sigma=1.0$). The trust factor absorbs partner stochasticity, keeping the policy anchored near the NBS.

\paragraph{Risk-averse RATTL and PPO.} Risk-averse RATTL ($\beta = -1.0$, Figure~\ref{fig:vary_noise_beta-1.0_rstd0.0_oppstd}) collapses to Hare (Stag rate ${\approx}0.10$) across all noise levels (Table~\ref{tab:ish_metrics_opp_noise}). Vanilla PPO (Figure~\ref{fig:vary_noise_vanilla_rstd0.0_oppstd}, same table) converges to a volatile mixed strategy (${\approx}75\%$ Stag) but fails to commit to full cooperation.

\begin{figure}[h]
     \centering
     \begin{subfigure}[b]{0.48\textwidth}
         \centering
         \includegraphics[width=\textwidth]{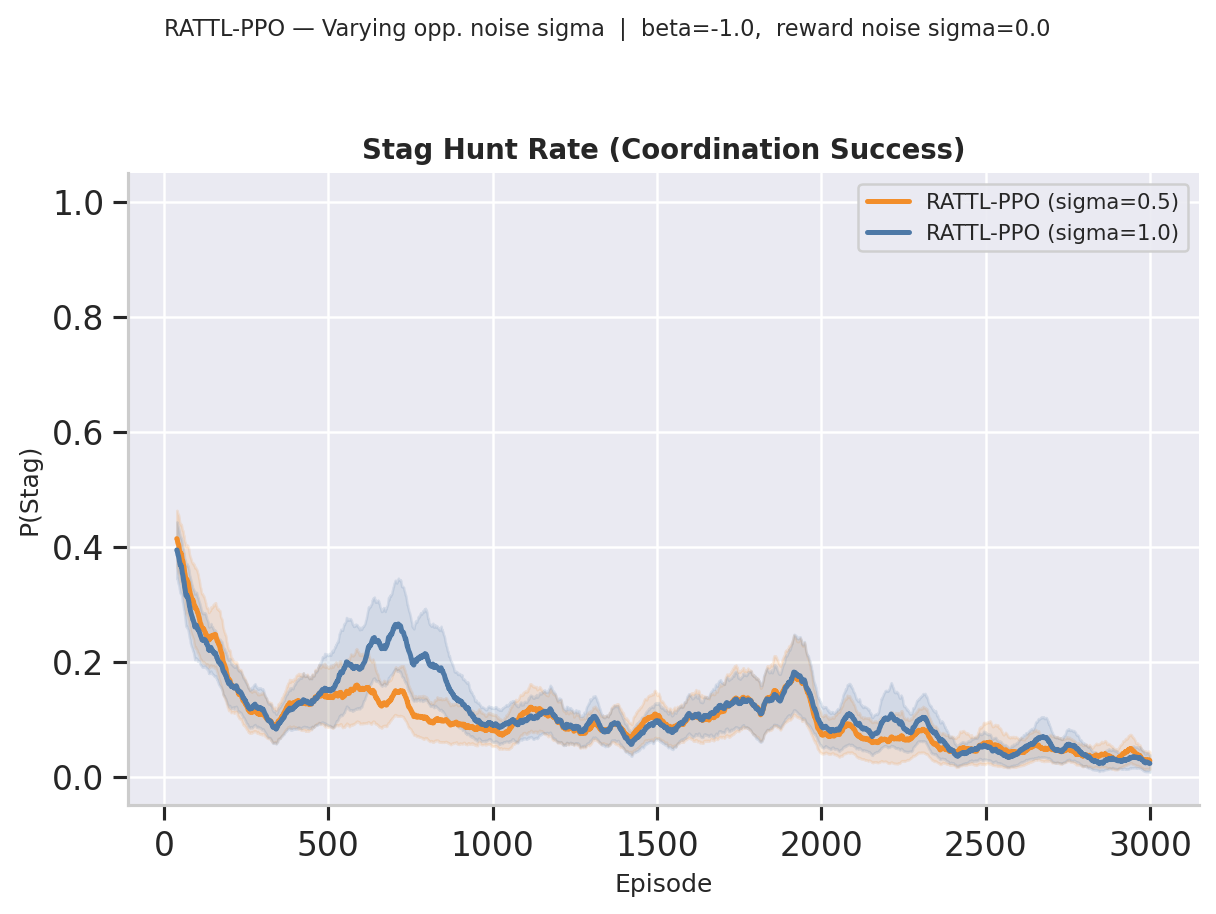}
         \caption{Attempt at coordination}
     \end{subfigure}
     \hfill
     \begin{subfigure}[b]{0.48\textwidth}
         \centering
         \includegraphics[width=\textwidth]{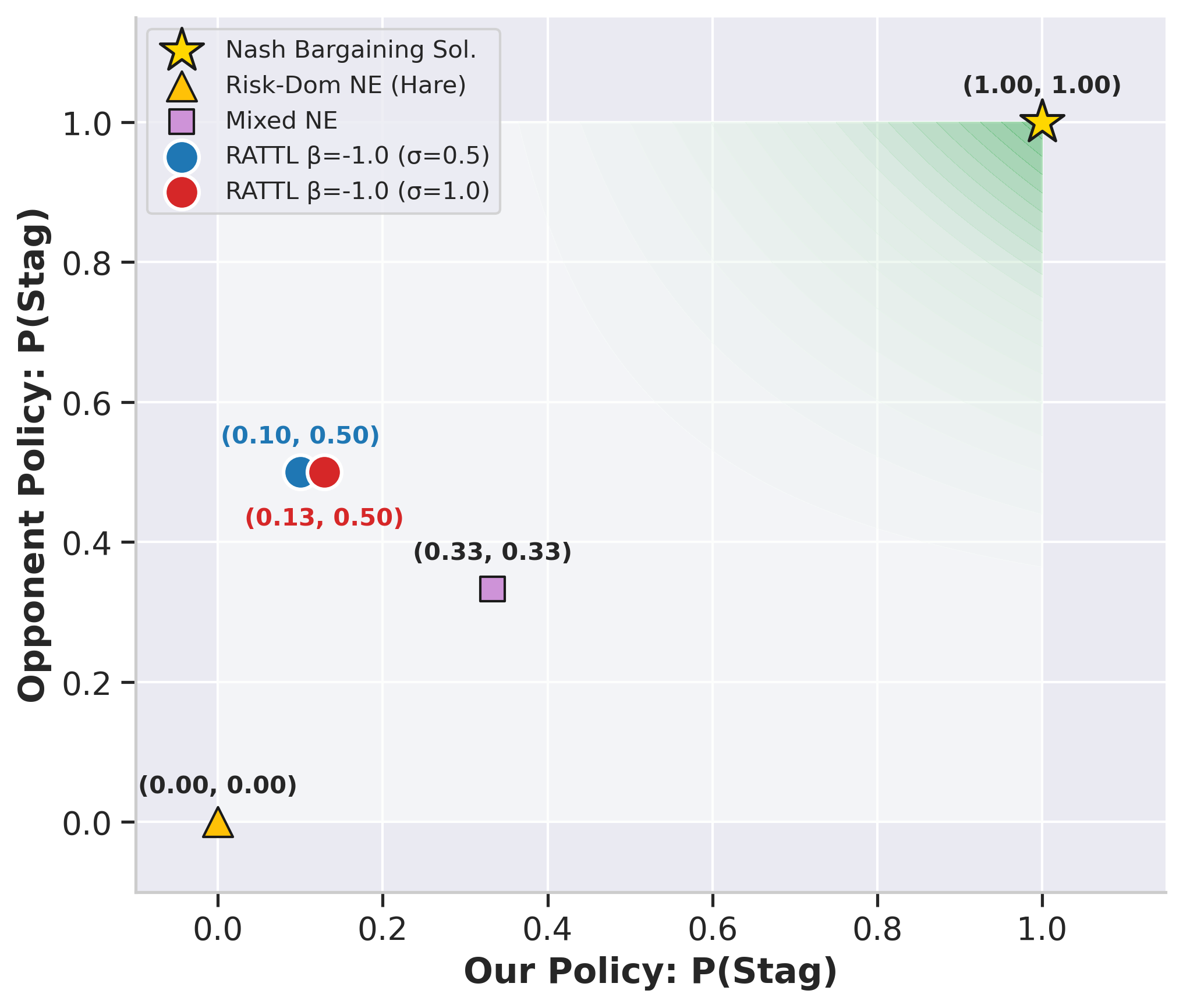}
         \caption{Mixed Strategy Nash Equilibria}
     \end{subfigure}
     \caption{Evaluation of RATTL-PPO ($\beta=-1.0$) against a stochastic opponent in the Iterated Stag Hunt where partner strategy is perturbed by standard-normal noise.}
     \label{fig:vary_noise_beta-1.0_rstd0.0_oppstd}
\end{figure}

\begin{figure}[h]
     \centering
     \begin{subfigure}[b]{0.48\textwidth}
         \centering
         \includegraphics[width=\textwidth]{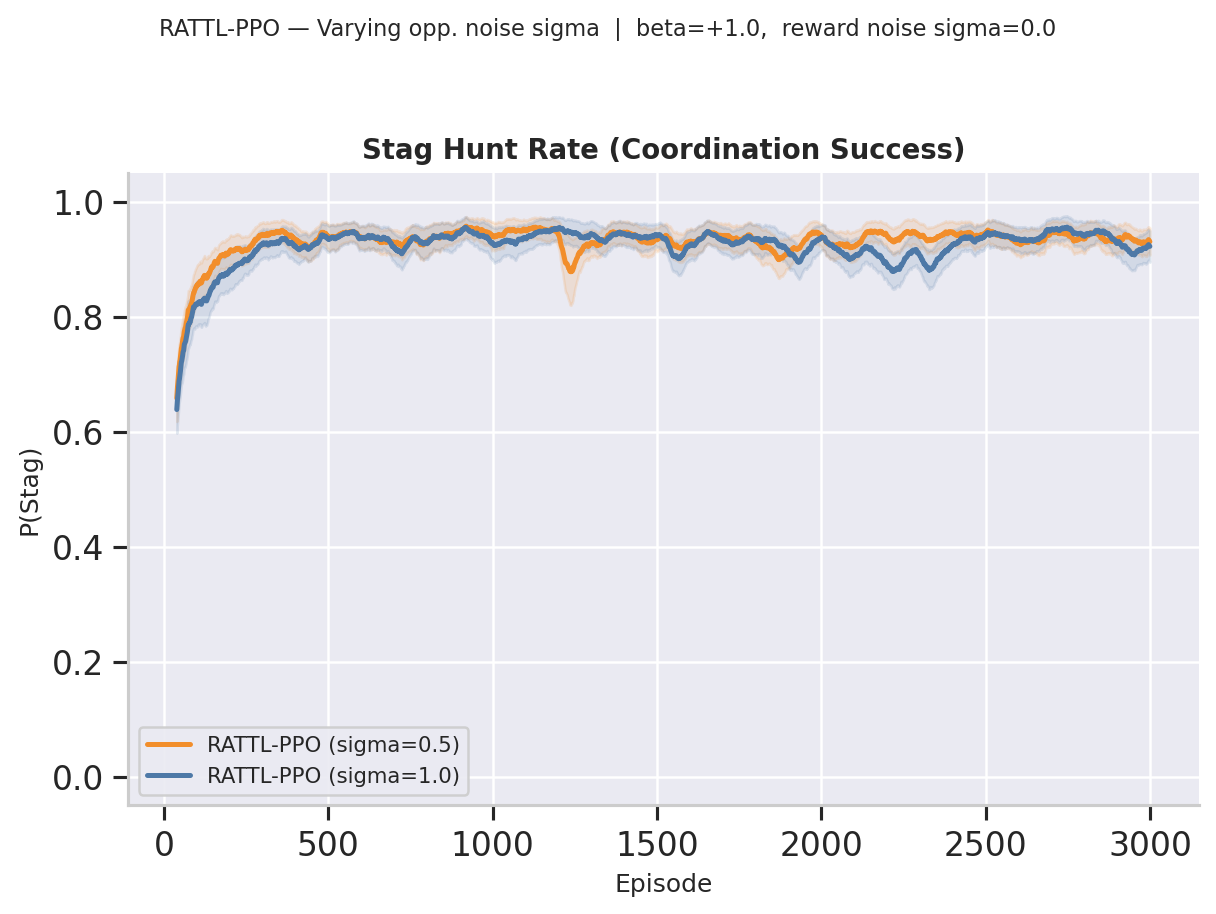}
         \caption{Attempt at coordination}
     \end{subfigure}
     \hfill
     \begin{subfigure}[b]{0.48\textwidth}
         \centering
         \includegraphics[width=\textwidth]{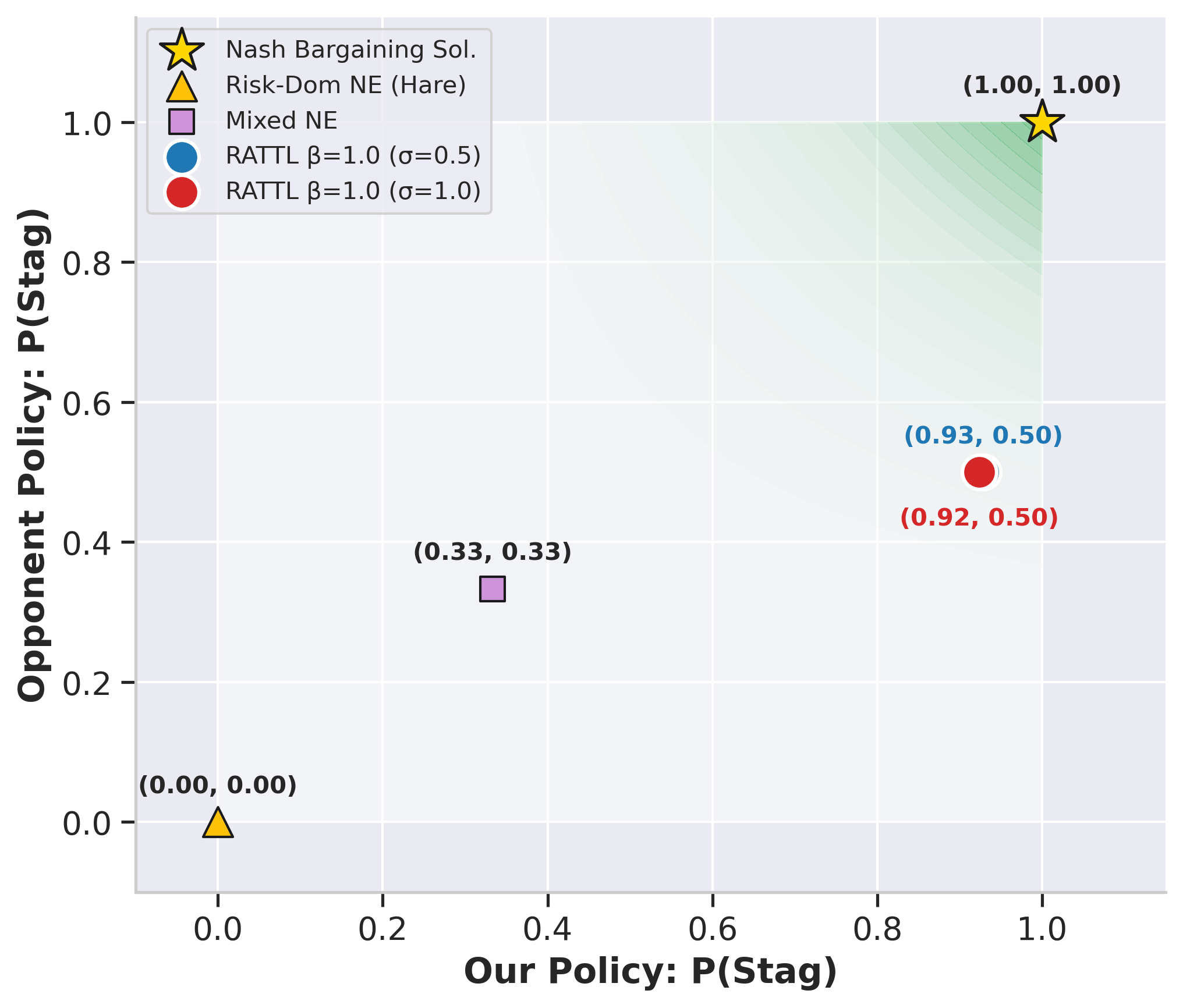}
         \caption{Mixed Strategy Nash Equilibria}
     \end{subfigure}
     \caption{Evaluation of RATTL-PPO ($\beta=1.0$) against a stochastic opponent in the Iterated Stag Hunt where partner strategy is perturbed by standard-normal noise.}
     \label{fig:vary_noise_beta1.0_rstd0.0_oppstd}
\end{figure}

\begin{figure}[h]
     \centering
     \begin{subfigure}[b]{0.48\textwidth}
         \centering
         \includegraphics[width=\textwidth]{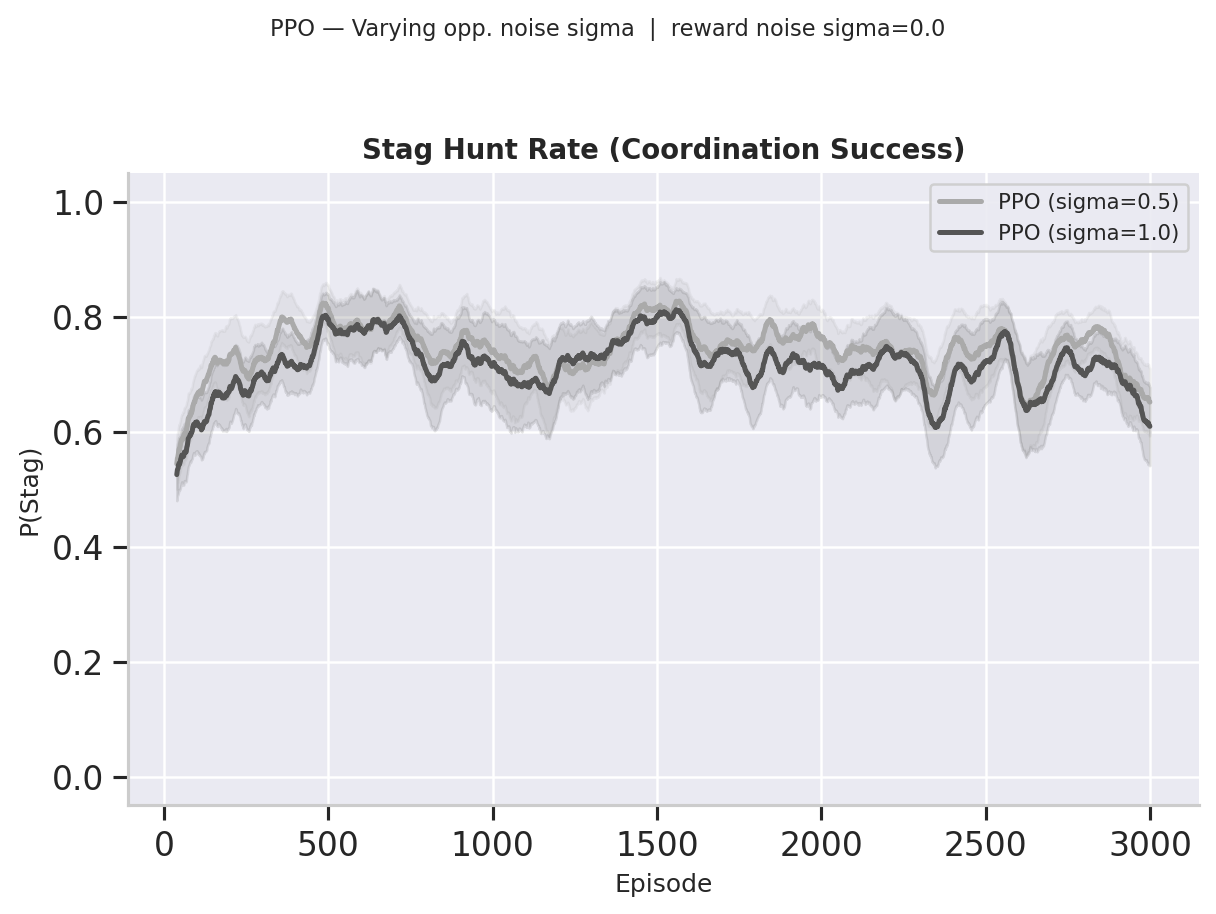}
         \caption{Attempt at coordination}
     \end{subfigure}
     \hfill
     \begin{subfigure}[b]{0.48\textwidth}
         \centering
         \includegraphics[width=\textwidth]{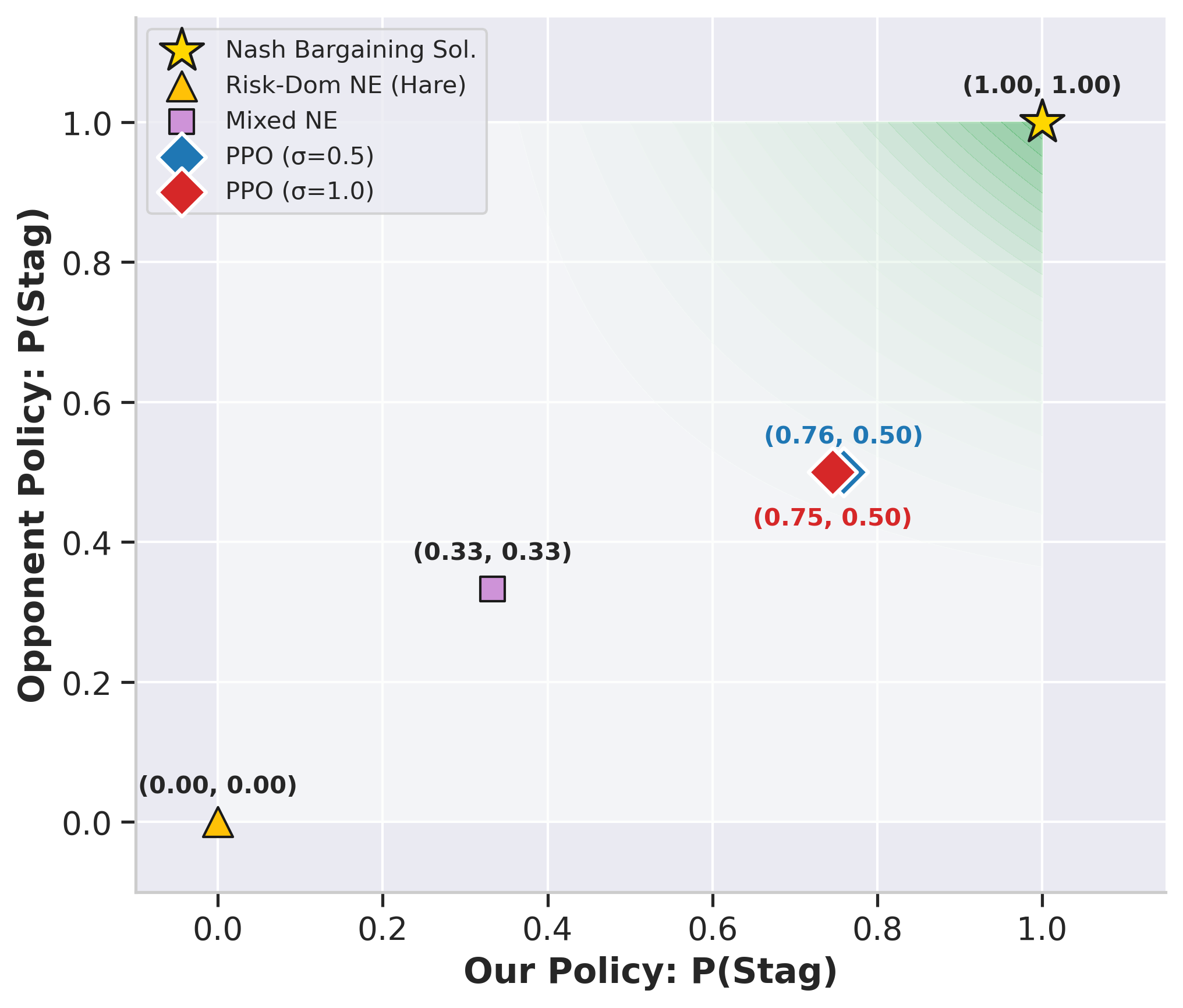}
         \caption{Mixed Strategy Nash Equilibria}
     \end{subfigure}
     \caption{Evaluation of vanilla PPO against a stochastic opponent in the Iterated Stag Hunt where partner strategy is perturbed by standard-normal noise.}
     \label{fig:vary_noise_vanilla_rstd0.0_oppstd}
\end{figure}

\begin{table}[h]
    \centering
    \caption{Empirical PoP and PoA across noise $\sigma$ added to the opponent's mixed strategy, for RATTL-PPO ($\beta=-1$), RATTL-PPO ($\beta=1$), and vanilla PPO.}
    \label{tab:ish_metrics_opp_noise}
    \begin{tabular}{@{}lcccccc@{}}
        \toprule
        & \multicolumn{2}{c}{RATTL ($\beta{=}{-}1$)} & \multicolumn{2}{c}{RATTL ($\beta{=}1$)} & \multicolumn{2}{c}{PPO} \\
        \cmidrule(lr){2-3}\cmidrule(lr){4-5}\cmidrule(lr){6-7}
        $\sigma$ & PoP$\uparrow$ & PoA$\downarrow$ & PoP$\uparrow$ & PoA$\downarrow$ & PoP$\uparrow$ & PoA$\downarrow$ \\
        \midrule
        0.5 & 1.16 & 4.59 & 2.86 & 1.86 & 2.53 & 2.11 \\
        1.0 & 1.27 & 4.23 & \textbf{2.88} & \textbf{1.85} & 2.51 & 2.13 \\
        \bottomrule
    \end{tabular}
\end{table}

\end{document}